\documentclass[
    pra,superscriptaddress, twocolumn, amsmath,amssymb, longbibliography
]{revtex4-2}
\usepackage{graphicx, color, xcolor}
\usepackage{dcolumn}
\usepackage{comment}
\usepackage{qcircuit} 
\usepackage{caption, subcaption}
\usepackage{braket}   
\usepackage{lmodern}
\usepackage{xr}
\usepackage{tikz}
\usepackage{relsize} 
\usepackage{amsmath, amsfonts, amssymb}
\usepackage{algorithm, algorithmic}
\usepackage{verbatim, multirow}
\usepackage{hyperref,cleveref}
\usepackage{float}
\usepackage[utf8]{inputenc}
\usepackage[T1]{fontenc}
\usepackage{booktabs}
\usepackage{soul}
\usepackage{mathtools}
\usepackage{setspace}
\usepackage{MnSymbol}
\usepackage{appendix,physics,wasysym,xparse}
\usepackage{bm,bbm,euscript,braket}

\newcommand{\cC}{\mathcal{C}}

\newcommand{\cE}{\mathcal{E}}

\newcommand{\cG}{\mathcal{G}}
\newcommand{\cH}{\mathcal{H}}

\newcommand{\cL}{\mathcal{L}}

\newcommand{\cO}{\mathcal{O}}

\newcommand{\cR}{\mathcal{R}}
\newcommand{\cS}{\mathcal{S}}

\hypersetup{colorlinks,linkcolor={red},citecolor={blue},urlcolor={blue}}

\newtheorem{Theorem}{Theorem}
\newtheorem{lemma}[Theorem]{Lemma}

\newtheorem{corollary}[Theorem]{Corollary}

\newenvironment{proof}{{\bf Proof:}}{\hfill$\square$}
\begin{document}
\preprint{APS/123-QED}

\title{Quantum Codes for Generalized Amplitude-damping Noise}

\author{Sourav Dutta}
\thanks{Corresponding author: sourav@physics.iitm.ac.in}
\affiliation{Department of Physics, IIT Madras, Chennai, India - 600036}
\affiliation{Center for Quantum Information, Communication and Computing, IIT Madras, Chennai, India - 600036}

\author{Anubhab Rudra}%
\affiliation{Department of Physics, IIT Madras, Chennai, India - 600036}
\affiliation{Center for Quantum Information, Communication and Computing, IIT Madras, Chennai, India - 600036}

\author{Manav Seksaria}
\affiliation{Department of Electrical Engineering, IIT Madras, Chennai, India - 600036}
\affiliation{Center for Quantum Information, Communication and Computing, IIT Madras, Chennai, India - 600036}

\author{Anil Prabhakar}
\affiliation{Department of Electrical Engineering, IIT Madras, Chennai, India - 600036}
\affiliation{Center for Quantum Information, Communication and Computing, IIT Madras, Chennai, India - 600036}

\author{Prabha Mandayam}
\affiliation{Department of Physics, IIT Madras, Chennai, India - 600036}
\affiliation{Center for Quantum Information, Communication and Computing, IIT Madras, Chennai, India - 600036}
\begin{abstract}

Quantum error correcting (QEC) plays a crucial role in protecting quantum information against decoherence and enabling scalable, reliable quantum computing. One of the most realistic and ubiquitous sources of noise affecting quantum hardware today
is generalized amplitude-damping (GAD) noise.
Conventional, deterministic QEC codes struggle to correct for GAD noise because of their inherent structure, leading to fidelity losses that scale linearly with the damping strength.
In this work, we introduce the framework of probabilistic approximate quantum error correction (PAQEC), that combines the flexibility of approximate QEC with the potential of post-selected recovery, enabling high-fidelity, resource-efficient error correction. We construct a five-qubit permutation-invariant code that, under probabilistic recovery, achieves a fidelity loss quadratic in the damping strength, thus outperforming existing QEC codes.
Formulating PAQEC as an optimization problem, we present a numerical technique based on Charnes--Cooper and semidefinite programming to identify the optimal recovery map for any PAQEC code.
Our results establish PAQEC as a powerful tool for developing resource-efficient, high-fidelity quantum codes tailored to realistic noise, with promising implications for near-term quantum devices and future fault-tolerant architectures.

\end{abstract}

\maketitle

\section{introduction}
Quantum computing~\cite{Feynman1982, Deutsch1985, Nielsen_Chuang_2010} promises to transform information processing by harnessing the quantum phenomena of superposition and entanglement to tackle problems beyond the reach of classical computers~\cite{Shor, Grover1996}, with impactful applications in cryptography~\cite{QC1_Yang2023, QC2_Akter}, materials science~\cite{MS1_Lordi2021}, drug discovery~\cite{Drug1_Flther2025}, and optimization~\cite{Opti1_Volpe2025, Opti2_Liu2024}.
However, building a large-scale quantum computer remains a formidable challenge owing to decoherence arising from continuous interaction with the environment, imperfect control operations, and measurement uncertainty.
Noise can manifest in various forms, such as bit-flips, phase-flips, depolarizing errors, or complex correlated errors, depending on the physical implementation and environment~\cite{Lidar_QEC_2013}. Without effective correction or mitigation, noise accumulates rapidly, disrupting quantum computations and rendering them unreliable.

Quantum error correction (QEC)~\cite{Shor1995,gottesman_thesis,Calderbank1996, Knill1998,Barbara_RevModPhys, QEC_zoo} 
provides a framework for protecting quantum information by encoding logical qubits into multiple physical qubits, enabling detection and correction of errors without directly measuring or disturbing the quantum state. General-purpose codes, such as the Shor code~\cite{Shor1995}, Steane code~\cite{Steane_code}, CSS codes~\cite{CSS_1, CSS_2} and surface codes~\cite{Kitaev2003, surface_code_1, surface_code_2}, 
are designed to correct for arbitrary Pauli errors, making them widely applicable across different quantum hardware platforms.
However, they may not be optimal for any specific noise model and are rather resource intensive. Noise-adapted QEC~\cite{akshaya_review}, on the other hand, tailors the encoding and decoding schemes to the dominant noise characteristics of a particular system, achieving high fidelity while requiring fewer resources -- see, for example~\cite{Leung_1997, Chuang_1997, fletcher2008, mg_2018, Dutta_2024}. This adaptability is particularly valuable in near-term quantum devices where the noise structure is known and often has a biased structure~\cite{aliferis2009, webster2015reducing}.

The search for good noise-adapted QEC schemes naturally leads to a generalization of the mathematical framework of QEC, beyond the well known Knill-Laflamme (KL) conditions~\cite{KLCondition}. These are a precise set of algebraic conditions for when a quantum code can perfectly correct a set of errors, which can be hard to satisfy for arbitrary noise models. Approximate quantum error correction (AQEC)~\cite{Leung_1997, beny, prabha} and probabilistic quantum error correction (PQEC)~\cite{pQEC1, pQEC2, pQEC3, pQEC4, dutta3qubit} solve this problem by extending the standard framework of quantum error correction to accommodate small deviations from the KL conditions.

By relaxing the constraint of perfect correctability, AQEC enables more flexible and efficient QEC schemes that are better suited to near-term quantum devices~\cite{debjyoti}. PQEC allows for the error correction process to be successful with only a certain probability, albeit exactly. By post-selecting over the successful instances, this approach can simplify decoding strategies and reduce overhead, thereby improving performance in practice~\cite{joshi2026}. Combining these two approaches can yield an interesting QEC framework that has not yet been studied in the literature.

Here, we construct a framework that combines AQEC and PQEC to design a probabilistic approximate quantum error correction (PAQEC) scheme that corrects for generalized amplitude-damping (GAD) noise. GAD\cite{GAD1, GAD2, GAD3, GAD4} is a finite-temperature generalization of amplitude-damping (AD) noise, and is one of the most realistic noise models in many quantum hardware. Previous studies on correcting GAD noise \cite{Cafaro_2014} have shown that deterministic codes that succeed with certainty cannot eliminate first-order errors.

In this work, we demonstrate the existence of a five-qubit permutation-invariant code that yields a fidelity loss quadratic in the damping strength, by adopting a probabilistic, approximate framework for QEC. The PAQEC framework enables approximate recovery of quantum states with high fidelity, probabilistically. We formulate algebraic conditions for PAQEC, and show that the problem of finding the optimal PAQEC recovery can be solved as a semidefinite program. 
 Similarly, we show that our algebraic conditions can also be used to find numerically optimized QEC codes that achieve a very high entanglement fidelity for arbitrary noise channels.

The rest of the paper is structured as follows.
In Sec.~\ref{sec:preli}, we briefly survey preliminary ideas on AQEC, PQEC and GAD noise.
In Sec.~\ref{sec:Dutta_5}, we introduce our five-qubit permutation-invariant (PI) code that can probabilistically correct for GAD noise with a fidelity loss that grows quadratically in damping strength. 
In Sec.~\ref{sec:PAQEC}, we formally introduce the most general error correction conditions, the PAQEC conditions, and discuss the performance of the five-qubit PI code under different recoveries.
In Sec.~\ref{sec:Optimization}, we use numerical techniques to both find the optimal recovery for PAQEC codes using semidefinite programming and also find good quantum codes that can correct the GAD noise. Finally, we conclude in Sec.~\ref{sec:conclusion} with a summary and future directions.

\section{Preliminaries}\label{sec:preli}
\subsection{Perfect and Approximate Quantum Error Correction (QEC)}

Quantum systems are fragile due to unavoidable interactions with their environment. A quantum noise channel can be modeled as a completely positive trace-preserving (CPTP) map($\mathcal{E}$), whose action on a state $\rho$, can be described by a set of Kraus operators $\{ E_m \}$, as, $\mathcal{E}(\rho) = \sum_m{E_m \rho E_m^\dagger}$ with $\sum_m{E_m^\dagger E_m} = I$. The central idea of quantum error correction (QEC) is to encode quantum information into a subspace $\mathcal{C}$ of the physical Hilbert space, spanned by orthonormal quantum states $\{ \ket{i_L} \}$, referred to as codewords. The \emph{codespace} $\mathcal{C}$ together with the corresponding encoding map defines a quantum error correcting code.

A quantum code with codewords $\{ \ket{i_L} \}$ can correct for a set of Kraus operators $\{ E_m \}_{m=1}^M$ \emph{perfectly} if and only if~\cite{KLCondition},
\begin{gather}\label{eq:KL}
    \bra{i_L}E_m^\dagger E_n \ket{j_L} = \lambda_{mn} \delta_{ij},
\end{gather}
where $\lambda_{mn}$ are complex scalars. 
These algebraic conditions, called the Knill-Laflamme (KL) conditions, guarantee that the errors $\{E_{i}\}$ map the codespace to mutually orthogonal subspaces, enabling perfect detection and reversal of the errors. The recovery operation in this case is simply a projective measurement -- the \emph{syndrome} measurement -- followed by a unitary operation~\cite{Nielsen_Chuang_2010}.

Approximate QEC relaxes the strictness of these conditions, allowing for shorter codes that can correct up to a certain order in noise strength. Although several variants of AQEC conditions exist in the literature, here we use the form from \cite{beny}, 
\begin{gather}\label{eq:KL_approx}
    \bra{i_L}E_m^\dagger E_n \ket{j_L} = \lambda_{mn} \delta_{ij} + \bra{i_L} B_{mn}\ket{j_L}.
\end{gather}
Here, $\{B_{mn}\}$ are arbitrary operators on the codespace, which characterize the deviation from the perfect QEC conditions in Eq.~\eqref{eq:KL}. If the perturbation operators $\{B_{mn}\}$ satisfy
\begin{align}
    \bra{i_L} B_{mn}\ket{j_L} \sim O(\epsilon^{t+1}),
\end{align}
uniformly for all $i,j,m,n$, then the KL conditions are satisfied up to order $O(\epsilon^{t+1})$. Then, it was shown in~\cite{beny} that there exists a recovery map, a CPTP map $\mathcal{R}$, that corrects the corresponding errors up to order $O(\epsilon^{t})$.

\subsection{Probabilistic QEC}

Probabilistic quantum error correction (PQEC) provides a framework for correcting noise using a recovery scheme that is not deterministic; rather, the recovery step involves post-selection. Here we briefly review the algebraic framework of PQEC developed in~\cite{dutta3qubit}. Given a quantum code $\cC = \{\ket{i_L}\}$, suppose the noise operators of the channel $\cE$ are grouped into $\mu$ sets of the form $\{E^{(a)}_m\}$, where $a \in [1, \mu],\ m \in [1, \eta_a]$, such that each set of $\eta_a$ error operators satisfies,
\begin{align}\label{eq:QEC_prob_cond}
    \sum_{v=1}^{\eta_a} \bra{i_L} E^{(a)\dagger}_v E^{(b)}_m\ket{j_L} = \chi^a_i  \delta_{ab} \delta_{ij}, \qquad \chi_i^a \neq 0.
\end{align}

Then, there exists a trace non-increasing recovery channel $\cR = \{R_a P_{a}\}$ that can perfectly correct the errors introduced by the noise channel $\cE$, where $\{P_a\}$ are a set of orthogonal projectors onto the noisy subspaces corresponding to the error set ${E^{(a)}_m}$ and $R_a$ is defined as~\cite{dutta3qubit}
\begin{align}\label{eq:reco_PQEC_perfect}
    R_a = \sum_{i,m}\frac{\lambda_a}{\chi^a_i} \ket{i_L} \bra{i_L} E^{(a)\dagger}_m.
\end{align}
Here, $\lambda_a$ is a normalizing factor chosen such that the largest eigenvalue of each $R_a^{\dagger} R_a$ is one, i.e., $R_a^{\dagger} R_a \leq I$.
This ensures that the channel $\cR$ remains trace non-increasing, since,
\begin{align}
    \sum_a \left(R_a P_a\right)^{\dagger} \left(R_a P_a\right) &= \sum_a P_a R_a^{\dagger} R_a P_a  \leq \sum_a P_a \leq I.
\end{align}

Such a recovery channel can be implemented via a unitary operation on an extended system, where successful implementation of the recovery is conditioned on the state of the ancilla system.

While the probabilistic nature of the correction introduces some post-selection induced sampling overhead, it can be advantageous in scenarios where exact recovery is more critical than deterministic operation, such as in near-term quantum devices or specialized quantum communication. Recently, the PQEC-based scheme to correct for amplitude-damping noise was successfully demonstrated on publicly available quantum hardware, with break-even performance~\cite{joshi2026}.

\subsection{Generalized amplitude-damping noise}
Amplitude-damping noise is one of the dominant noise models in present-day quantum hardware, that arises from a Jaynes–Cummings interaction~\cite{Shore1993, Jaynes1963} between the quantum system and its environment, which is assumed to be at the absolute zero temperature. The state $\ket{0}$ is the unique fixed point of the AD channel, whereas $\ket{1}$ decays to $\ket{0}$ with probability $\gamma$.
The Kraus operators for the AD noise for qubit systems are thus written as,
\begin{align}
    \begin{split}
        & A_0 = \begin{pmatrix}
        1 & 0 \\ 0 & \sqrt{1-\gamma}
    \end{pmatrix}, ~~ A_1 = \begin{pmatrix}
        0 & \sqrt{\gamma} \\ 0 & 0 \end{pmatrix},
    \end{split}
    \label{eq:ad_Kraus}
\end{align}
where $\gamma$ is the damping strength of the noise-channel.

Generalized amplitude-damping noise~\cite{Cafaro_2014} is the finite-temperature generalization of AD noise~\cite{Chuang_1997, mg_2018, Dutta_2024}.
This noise model assumes that the environment is not at absolute zero but at some finite temperature, leading to a finite probability of excitation from $\ket{1}$ to $\ket{0}$along with the damping from $\ket{1}$ to $\ket{0}$.

For a two-level system, the Kraus operators of a GAD channel are given by
\begin{align}
    \begin{split}
        & A_0 = \sqrt{1 - p} \begin{pmatrix}
        1 & 0 \\ 0 & \sqrt{1-\gamma}
    \end{pmatrix}, ~~ A_1 = \sqrt{1 - p} \begin{pmatrix}
        0 & \sqrt{\gamma} \\ 0 & 0 \end{pmatrix}, \\
    & B_0 = \sqrt{p} \begin{pmatrix}
        \sqrt{1-\gamma} & 0 \\ 0 & 1
    \end{pmatrix}, ~~~~ B_1 = \sqrt{p} \begin{pmatrix}
        0 & 0 \\ \sqrt{\gamma} & 0 \end{pmatrix}.
    \end{split}
    \label{eq:gad_Kraus}
\end{align}

\noindent The error operator $A_{1}$ represents a decay process that maps $\ket{1}$ to $\ket{0}$ and annihilates the state $\ket{0}$. The operator $B_1$ represents the excitation process; it is associated with the absorption of a photon from the environment into the system, unlike $A_1$, which is associated with the emission of a photon from the system to the environment.

The parameter $p$ which representation the relative probability of occurrence of the excitation process over the relaxation, depends on the temperature $T_0$ of the environment through the Boltzmann distribution $$p = \frac{e^{-\frac{\hbar \omega}{k_B T_0}}}{1 + e^{-\frac{\hbar \omega}{k_B T_0}}}.$$ 
Here, $\hbar$, $\omega$, and $k_B$ denote the reduced Planck constant, the transition angular frequency of the qubit, and the Boltzmann constant, respectively, while $T_0$ is the temperature of the environment. At absolute zero ($T_0 = 0$ K), the parameter $p$ equals $0$, implying that the excitation process is completely suppressed. At high temperatures ($T_0 \rightarrow \infty$ ), the parameter $p$ converges to the value $\frac{1}{2}$~\cite{Cafaro_2014}, indicating that the excitation and decay processes are now equally likely.
\subsection{Benchmarking {PQEC: Fidelity and Success Probability}}
One of the standard measures to quantify the performance of a QEC scheme is the entanglement fidelity. For a code $\mathcal{C}$ with codewords $\{ \ket{i_L} \}$, the entanglement fidelity after the action of noise channel $\mathcal{E}$ and recovery $\mathcal{R}$ is defined as~\cite{Nielsen_Chuang_2010},
\begin{equation}
F_{\text{ent}} = \bra{\Psi} (\cR \circ \cE) \otimes I(\ket{\Psi}\bra{\Psi})\ket{\Psi}, \label{eq:ent_fid}
\end{equation}
where $\ket{\Psi}$ is a purification of the maximally mixed state on the codespace $\rho_{L} = \frac{1}{2^k} \sum_{i=0}^{2^k -1}\ket{i_L}\bra{i_L}$, given by,
\begin{equation}\label{eq:pure}
\ket{\Psi} = \frac{1}{\sqrt{2^k}} \sum_{i=0}^{2^k - 1} \ket{i_L} \otimes \ket{i}.
\end{equation}
Unlike state fidelity, which quantifies the performance of the QEC protocol for individual states, the entanglement fidelity serves as a lower bound of the average state fidelity~\cite{Schumacher_Nielsen_Entanglement_fidelity}.

For our post-selected QEC setting, where the recovery is probabilistic, the entanglement fidelity takes the form
\begin{equation}\label{eq:ent_fid_prob_case}
    F_{\text{ent}} = \frac{\bra{\Psi} (\cR \circ \cE) \otimes I(\ket{\Psi}\bra{\Psi})\ket{\Psi}}{\text{Tr}[ (\cR \circ \cE) \otimes I(\ket{\Psi}\bra{\Psi})]},
\end{equation}
where we discard the undesired post-selection outcomes and calculate the fidelity only for the successful instances. {We further prove in Appendix \ref{app:A} that the Haar-averaged probability of successful recovery is given by
\begin{equation}
\Bar{p}_{success} = \Tr[((\cR \circ \cE)\otimes I)(\dyad{\Psi})], \label{eq:avg_prob}
\end{equation}
where $|\Psi\rangle$ is the maximally entangled logical state defined in Eq.~\eqref{eq:pure}.}

\section{A $5$-qubit code for generalized amplitude-damping noise}\label{sec:Dutta_5}
Correcting GAD noise to first order in the damping strength poses a significant challenge for conventional codes. This difficulty arises from the structure of the GAD noise channel -- while the damping and excitation errors, denoted by $A_1$ and $B_1$, can be detected with the appropriate choice of code, it is much harder to identify the nature of the no-damping errors $A_0$ and $B_0$.

Indeed, conventional QEC schemes do not correct for first-order GAD errors completely, and the fidelity loss contains a term of the form $p\gamma$ in it~\cite{Cafaro_2014}. Typically, the value of $p$ is assumed to be much less than $\gamma$, and the term $p\gamma$ is considered to be equivalent to $\cO(\gamma^2)$. However, different physical gates can change the structure of this noise. For example, under the action of the Pauli $X$ gate, $XA_0 = B_0 X$ and $XA_1 = B_1X$. In other words, the $X$ gate does not change the inherent structure of the Kraus operators of the GAD noise but changes their probabilities of occurrence; that is, the errors $\{A_m\}$ in Eq.~\eqref{eq:gad_Kraus} now occur with probability $p$ and the errors $\{B_m\}$ occur with probability $1-p$. Hence, the fidelity loss term $p\gamma$ can no longer be associated with an $\cO(\gamma^2)$ term.

The existence of a quantum code that yields a fidelity loss that is genuinely of order $\cO(\gamma^2)$ under GAD noise remains an open problem. We now show that such a code is indeed possible in the probabilistic approximate QEC (PAQEC) framework.

We first group the GAD channel errors that are first order in $p$ and $\gamma$, into three sets. The first set, denoted as $\cG^{(0)}$, consists of all the tensor-product combinations of the no-damping errors $A_0$ and $B_0$.
The second (third) set of errors, denoted as $\cG^{(\downarrow)}$ ($\cG^{(\uparrow)}$), consists of a single damping (excitation) error $A_1$ ($B_1$) on one of the physical qubits, while the rest of the qubits are subject to combinations of no-damping errors $A_0$ and $B_0$.
Our aim is to approximately correct these three sets of errors such that the leading order fidelity-loss is $\cO(\gamma^2)$.

\begin{figure}
    \centering
    \includegraphics[width=0.99\linewidth]{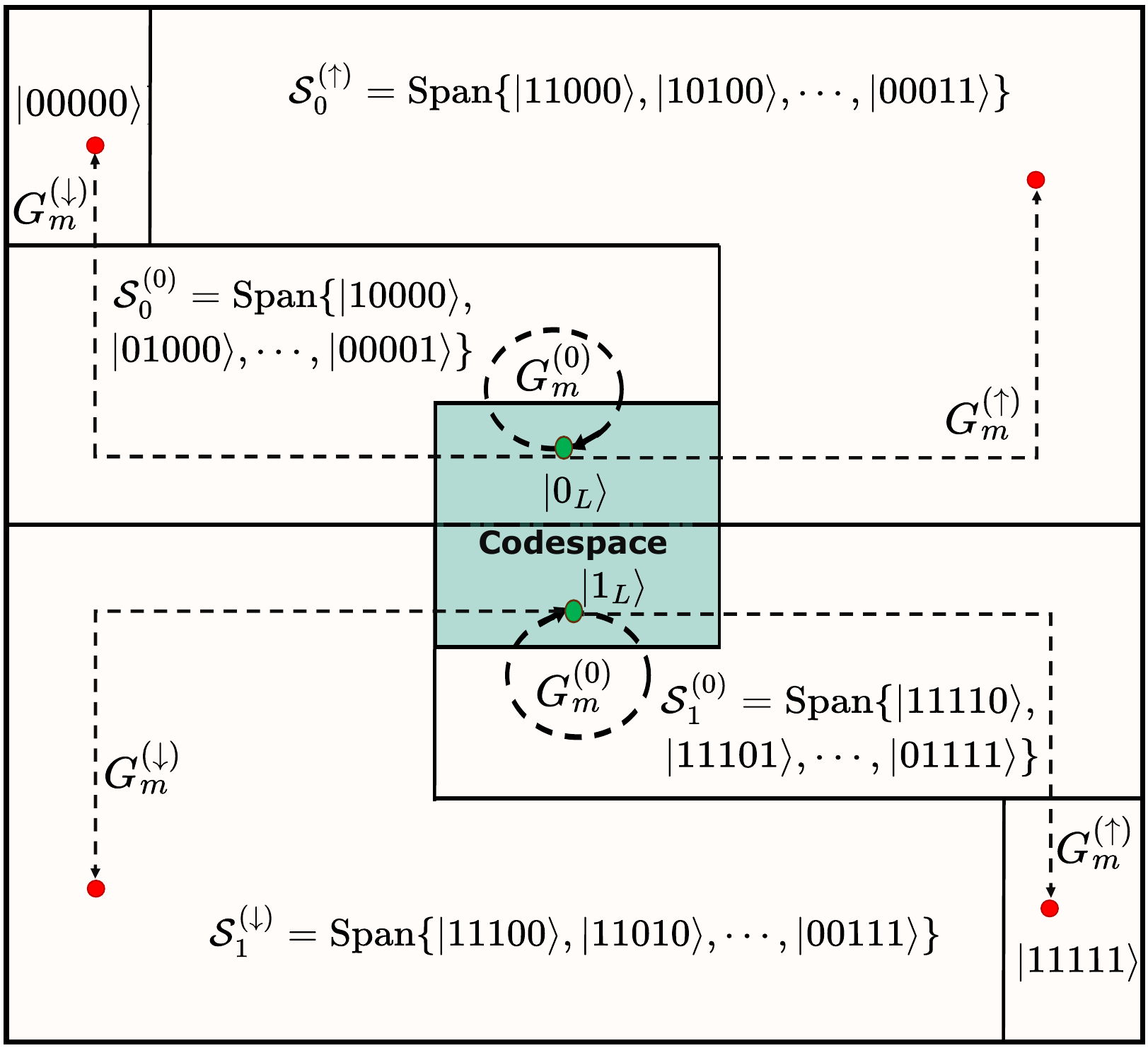}
    \caption{\textbf{Action of generalized amplitude-damping on the $[5,1]$ PI code.}
A first-order GAD error $G_m^{(a)}$ maps codeword $\ket{i_L}$ into a state in the subspace $\cS_i^{(a)}$. Subspaces $\{\cS_i^{(a)}\}_{i,a}$ are mutually orthogonal for different $i$ and $a$. Note that the codespace overlaps with $\cS_0^{(0)}$ and $\cS_1^{(0)}$.}
    \label{fig:Dutta5_HS}
\end{figure}

\subsection{The $[5,1]$ Permutation-Invariant Code}\label{subsec:encoding}
We encode the physical qubits into permutation-invariant (PI) states with fixed excitation numbers, also known as Dicke States~\cite{dicke}. 
We note that our quantum code for GAD is distinct from the known examples of QEC codes constructed based on PI states~\cite{Ouyang2014, Ouyang2026, Chandra2026, dutta3qubit}.
We encode the logical ``$0$'' state in a five-qubit Dicke state with excitation number $1$ and encode the logical ``$1$'' in the five-qubit Dicke state with excitation number $4$, as shown below.
\begin{align} \label{eq:code_dutta5}
    \begin{split}
        \ket{0_L} = \frac{1}{\sqrt{5}} ( \ket{10000} + \ket{01000} + \ket{00100} + \ket{00010} + \ket{00001} ), \\
        \ket{1_L} = \frac{1}{\sqrt{5}} ( \ket{11110} + \ket{11101} + \ket{11011} + \ket{10111} + \ket{01111} ).
    \end{split}
\end{align}

For a five-qubit code, a total of $2^5 = 32$ error operators are present in $\cG^{(0)} $ set, which is represented as $\{G^{(0)}_m\}_{m=1}^{32}$.
The second set $\cG^{(\downarrow)}$ consists of errors with one damping error $A_1$ on one of the five qubits and a combination of $A_0$ and $B_0$ acting on the rest. Finally, the third set $\cG^{(\uparrow)}$ consists of a single excitation error $B_1$ on one qubit and combinations of $A_0$ and $B_0$ acting on the rest. The number of errors in each of the sets $\cG^{(\downarrow)}$ and $\cG^{(\uparrow)}$ is $5 \times 2^{5-1} = 80$ and they are indexed as $G^{(\downarrow)}_m$ and $G^{(\uparrow)}_m$ respectively.

The action of first-order GAD errors on the codewords in Eq.~\eqref{eq:code_dutta5} is schematically illustrated in Fig. \ref{fig:Dutta5_HS}. The no-damping errors $\mathcal{G}^{(0)}$ merely distort the codewords by a $\gamma$-dependent scalar factor.
For example, the error operator $A_{00000} = A_{0}\otimes A_{0}\otimes A_{0}\otimes A_{0}\otimes A_{0}$ acts as,
\begin{align}
    \begin{split}
        &A_{00000}|0_L\rangle = \sqrt{1-\gamma} |0_L\rangle, ~~ A_{00000}|1_L\rangle = (1-\gamma)^2 |1_L\rangle
    \end{split}
\end{align}
However, errors that contain both $A_0$ and $B_0$ will distort the codewords, as exemplified by the action of the error operator $A_{000}B_{00}$ $= A_{0}\otimes A_{0}\otimes A_{0}\otimes B_{0}\otimes B_{0}$ below.
\begin{align}
    \begin{split}
        A_{000}B_{00}|0_L\rangle &= \frac{1}{\sqrt{5}}\left( (1-\gamma)^{\frac{3}{2}} (|10000\rangle + |01000\rangle+ |00100\rangle) \right. \\
        & \left. + (1-\gamma)^{\frac{1}{2}} (|00010\rangle + |00001\rangle) \right) \\
        A_{000}B_{00}|1_L\rangle &= \frac{1}{\sqrt{5}}\left( (1-\gamma) (|01111\rangle + |10111\rangle+ |11011\rangle) \right. \\
        & \left. + (1-\gamma)^2 (|11101\rangle + |11110\rangle) \right)
    \end{split}
\end{align}

\noindent In contrast, the errors in $\mathcal{G}^{(\downarrow)}$ map the codewords $\ket{0_L}$ and $\ket{1_L}$ to states with excitation numbers $0$ (that is, $\ket{00000}$) and $3$ (for example, $\ket{11100}$ or $\ket{10101}$), respectively. The errors in $\cG^{(\uparrow)}$ on the other hand, map the codewords $\ket{0_L}$ and $\ket{1_L}$ to states with excitation number $2$ (for example, $\ket{11000}$ or $\ket{00110}$) and $5$ (that is, $\ket{11111}$).

To summarize, a first-order GAD error $G_m^{(a)}$ takes codeword $\ket{i_L}$ into a state in the subspace $\cS_i^{(a)}$ such that the \emph{noisy subspaces} $\{\cS_i^{(a)}\}_{i,a}$ are mutually orthogonal for different $i$ and $a$ (see Fig.~\ref{fig:Dutta5_HS}). Note that the codespace itself overlaps with $\cS_0^{(0)}$ and $\cS_1^{(0)}$. A direct consequence of these properties is that we can perform a projective measurement to unambiguously differentiate between the different groups of single-qubit errors.

The projectors $\{P_{0}, P_{\downarrow}, P_{\uparrow}\}$ onto the error subspaces associated with the error sets $\{\mathcal{G}^{(0)}, \mathcal{G}^{(\downarrow)}, \mathcal{G}^{(\uparrow)}\}$ can be expressed as sums of computational-basis projectors
$\dyad{x}$ labeled by the excitation number of the state, or equivalently, the Hamming weight $\mathrm{wt}(x)$,
where $x \in \{0,1\}^{\times 5}$.
\begin{align}
P_0 &= \sum_{\mathrm{wt}(x)=1,4} \dyad{x}, \nonumber\\
P_{\downarrow} &= \sum_{\mathrm{wt}(x)=0,3} \dyad{x}, \nonumber\\
P_{\uparrow} &= \sum_{\mathrm{wt}(x)=2,5} \dyad{x}.
\end{align}
Next, we present a compact quantum circuit that performs the syndrome measurement corresponding to these projectors and unambiguously differentiates between the no-damping, single-damping, and single-excitation errors.

\subsection{Circuit for syndrome measurement}

\begin{figure*}
    \centering
    \includegraphics[width=0.95\linewidth]{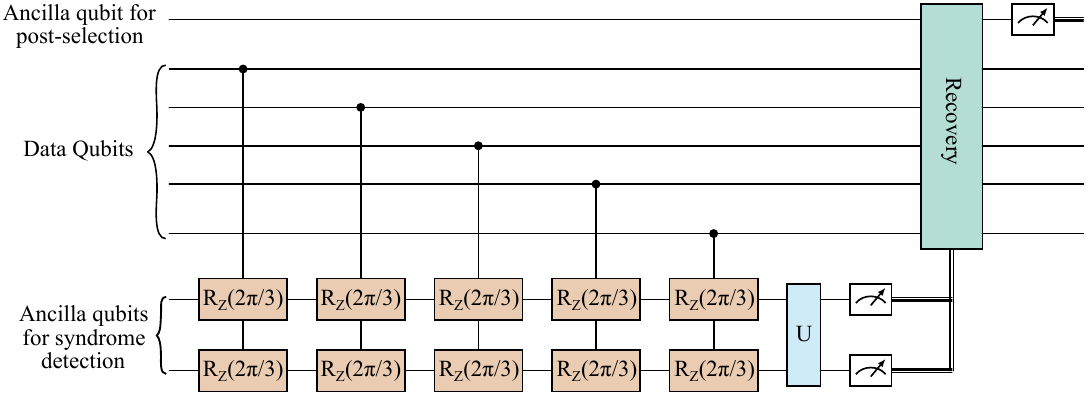}
    \caption{\textbf{Syndrome-based recovery protocol for the $[5,1]$ PI code. } Circuit for syndrome-based recovery of the five-qubit permutation-invariant code under GAD noise.}
    \label{fig:syn}
\end{figure*}

The $[5,1]$ PI code is not a stabilizer code, as it does not admit four independent stabilizer generators.
However, there exist two \emph{symmetry} operators that act non-trivially on the logical states, either of which can be used to determine the type of error (no damping, single excitation, or single damping) that occurred.
Consider the operators,
\begin{equation}\label{eq:stabs}
    S_1 = R_Z\left(\frac{2\pi}{3}\right)^{\otimes 5} \qquad S_2 = R_Z\left(-\frac{\pi}{3}\right)^{\otimes 5},
\end{equation}
where $R_{Z}(\theta) = e^{i\frac{\theta}{2}Z}$, is a rotation by angle $\theta$ about the $Z$-axis of the Bloch sphere. These operators commute with the no-damping errors, as they are all diagonal matrices. For the single damping errors $A_m$ and single excitation errors $B_m$, we have
\begin{equation}
    A_m S_1 = \omega S_1 A_m, \qquad B_m S_1 = \omega^2 S_1 B_m,
\end{equation}
where $\omega$ is the third root of unity.
Given these commutation relations, it is impossible to detect the type of error using a single qubit ancilla; rather, it requires a three-dimensional quantum system, that is, a qutrit ancilla.

Instead of a qutrit, we use two ancilla qubits with the states $\ket{00}$, $\ket{01}$, and $\ket{11}$ representing the three levels of a qutrit. We first prepare the ancilla qubits in the state,
$$\ket{\psi} = \frac{1}{\sqrt{3}} (\ket{00} + \ket{01} + \ket{11}).$$ We then perform the following controlled operation,
$$ CR_{ZZ} = \dyad{0} \otimes I^{\otimes 2} + \dyad{1} \otimes R_Z\left(\frac{2\pi}{3}\right)^{\otimes 2}, $$
with the data qubits acting as the control and the target being the ancilla qubits, as shown in Fig. \ref{fig:syn}. If there are no-damping errors, the state of the ancilla qubits can be written as,
\begin{equation}
    \ket{\psi_{0}} = \frac{1}{\sqrt{3}} (\ket{00} + \omega \ket{01} + \omega^2 \ket{11}),
\end{equation}
whereas in the case of a single-damping error or a single-excitation error, the states of the ancilla qubits are respectively,
\begin{align}
    \ket{\psi_{\downarrow}} &= \frac{1}{\sqrt{3}} (\ket{00} + \ket{01} + \ket{11}), \\
    \ket{\psi_{\uparrow}} &= \frac{1}{\sqrt{3}} (\ket{00} + \omega^2 \ket{01} + \omega \ket{11}).
\end{align}
Note that the states $\ket{\psi_0}, \ket{\psi_{\downarrow}}$ and $\ket{\psi_{\uparrow}}$ are orthogonal to each other. We can distinguish between these three states by measuring the ancilla qubits in an appropriate basis. This completes the syndrome measurement, which reveals the type of error that occurred.

\subsection{Recovery Circuit}\label{subsec:recovery}

Once the type of error has been identified, we perform the appropriate recovery operation on the encoded qubits. Unlike the case of Pauli errors and stabilizer codes, the recovery operators for this scheme are non-unitary operators on the five-qubit space. Such a recovery operator can be implemented as a unitary on six qubits, followed by a measurement and post-selection on the additional ancilla qubit for successful implementation of the recovery~\cite{non_unitary_2005, dutta3qubit}.

Since the errors $G^{(a)}_m$, corresponding to different error sets $\mathcal{G}^{(a)}$, map the logical states to orthogonal error subspaces $S_i^{(a)}$, as shown in  Fig. \ref{fig:Dutta5_HS}, a natural recovery strategy is to construct a recovery operator that projects any state from $S^{(a)}_i$ to logical state $|i_L\rangle$. We need to prefix these operators with coefficients $g^a_i(\gamma,p)$ to balance the contribution of different error operators according to their likelihood under the noise model.

The recovery for the $[5,1]$ PI code can thus be implemented using operators of the form,
\begin{equation}\label{eq:reco}
    R_a = \sum_{i,m} g_{i}^{a}(\gamma, p) \dyad{i_L} G^{(a) \dagger}_m,
\end{equation}
where $a \in \{0, \uparrow, \downarrow\}$ and $g_i^a(\gamma,p)$ are functions of the noise parameters $\gamma$ and $p$ chosen in such a way that the recovery channel comprising the operators $\{R_{a}\}$ is trace non-increasing. This recovery procedure is similar in structure to that defined by the operators in  Eq.~\eqref{eq:reco_PQEC_perfect}, and maybe implemented via a unitary on an extended space, with a probability of success that depends on the noise strength ~\cite{dutta3qubit}. Additionally, since our syndrome cannot identify the exact error from the error set $\cG^{(a)}$, $R_a$ can only approximately correct for the errors in the group $\cG^{(a)}$. Our $[5,1]$ PI code is thus an example of a probabilistic and approximate QEC code for GAD noise.

We can calculate the entanglement fidelity of our $5$-qubit QEC scheme in the presence of GAD noise
using the expressions in Eq.~\eqref{eq:ent_fid}, as,
\begin{align}\label{eq:ent_fid_23}
    F_{\text{ent}} &= \frac{\sum_{a,m} \bra{\Psi}(R_a G^{(a)}_m \otimes I) \dyad{\Psi} (R_a G^{(a)}_m \otimes I)^{\dagger}\ket{\Psi}}{\text{Tr}\left[\sum_{a,m}(R_a G^{(a)}_m \otimes I) \dyad{\Psi} (R_a G^{(a)}_m \otimes I)^{\dagger}\right]}\nonumber\\
    &= \frac{1}{2} \frac{\sum_{am} \left(\sum_{i} g^a_i(\gamma,p) \sum_{\nu}\bra{i_L}G^{(a)\dagger}_{\nu} G^{(a)}_m\ket{i_L}\right)^2}{\sum_{ami} \left(g^a_i(\gamma, p) \sum_{j\nu} \bra{i_L} G^{(a)\dagger}_{\nu} G^{(a)}_{m}\ket{j_L}\right)^2}.
\end{align}
We still need to find a suitable form for $g_i^a(\gamma,p)$ which maximizes this fidelity, while keeping the recovery channel trace non-increasing.
In our case, we define $g_i^a(\gamma,p)$ to be the inverse of the average of $\sum_{\nu}\bra{i_L}G^{(a)\dagger}_{\nu} G^{(a)}_m\ket{i_L}$ over $m$, that is,
\begin{align}\label{eq:gai}
    g_i^a(\gamma,p) = \lambda_a \left(\frac{1}{|m|} \sum_{m,\nu}\bra{i_L}G^{(a)\dagger}_{\nu} G^{(a)}_m\ket{i_L}\right)^{-1},
\end{align}
with an additional {free parameter $\lambda_a$} to take care of the trace non-increasing nature of the channel. Such a choice leads us to an entanglement fidelity {that is truly second order in both $\gamma$ and $p$, for the $[5,1]$ PI code}.
\begin{align} \label{eq:ent_fid_avg}
    F_{\text{ent}} =& 1 - \left(\frac{27}{10}p -9 p^{\frac{3}{2}} + \frac{153}{10}p^2 - \frac{45}{2}p^{\frac{5}{2}} + \cO(p^3) \right) \gamma^2 + \cO(\gamma^3).
\end{align}
\begin{figure*}[t!]
    \centering
    \includegraphics[width=1\linewidth]{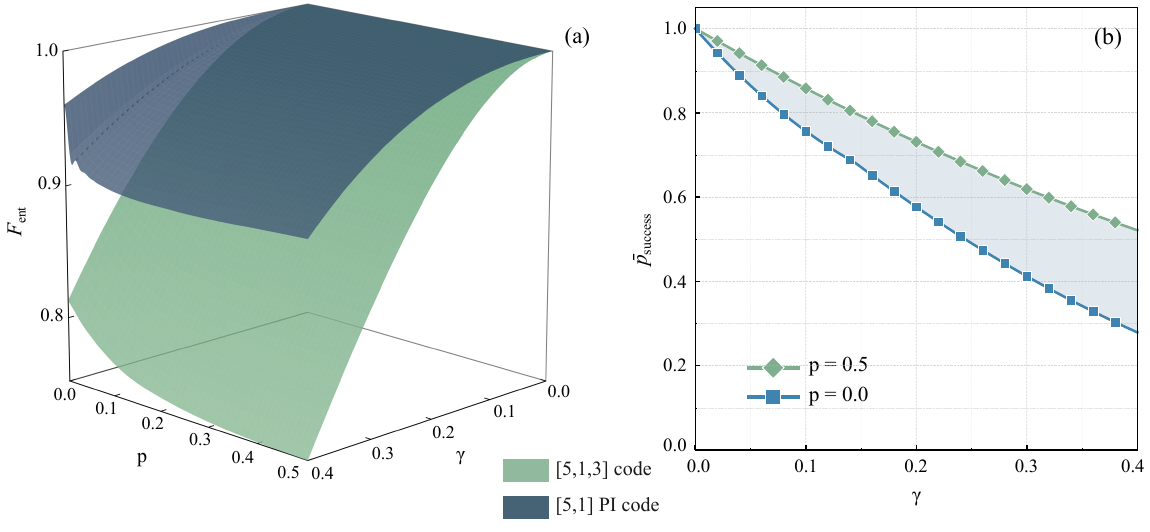}
    \caption{\textbf{Performance of $[5,1]$ PI code.}
    (a) Entanglement fidelity as a function of noise parameters $\gamma$ and $p$, for the $[5,1]$ PI code and the $[5,1,3]$ stabilizer code with a numerically optimized recovery.
    (b) {Average} probability of successful recovery as a function of $\gamma$, with {$0 \leq p \leq 0.5$}, for the $[5,1]$ PI code.}
    \label{fig:comp}
\end{figure*}
\noindent In fact, if we consider $p$ to be of the same order as $\gamma$, Eq. \eqref{eq:ent_fid_avg} suggests that the leading order of entanglement fidelity loss is effectively $\mathcal{O}(\gamma^3)$.

In Fig.~\ref{fig:comp} (a), we compare the performance of the $[5,1]$ PI code with that of the $[5,1,3]$ stabilizer code using a numerically optimized recovery map obtained via a semidefinite program (SDP)~\cite{Fletcher_sdp}. We see that our $[5,1]$ PI code outperforms the $[5,1,3]$ stabilizer code for all the values of $\gamma$ and $p$.

The probability of successfully implementing the recovery channel described by the operators in Eq.~\eqref{eq:reco}, is derived in Sec.~\ref{sec:analytical} below. Fig.~\ref{fig:comp} (b) shows how the probability of successfully implementing the recovery operation changes with the strength of the noise. This numerical plot is based on the analytical expression for the success probability obtained in Eq.~\eqref{eq:prob_succ}.

\subsection{Logical gates}\label{sec:logical_gates}

Going beyond the QEC scheme, we also obtain {a universal set of} logical gates for the $[5,1]$ PI code.
We observe that the code defined in Eq. \eqref{eq:code_dutta5} has a symmetry which naturally leads to transversal logical $\overline{X}$ and $\overline{R}_Z(\theta)$ gates, as given below.
\begin{align}
    &\overline{X} = X \otimes X \otimes X \otimes X \otimes X \\
    &\overline{R}_Z(\theta) = e^{-\frac{i\theta}{3}} \left(R_Z\left(\frac{\theta}{3}\right)\right)^{\otimes 5} \label{eq:logical}
\end{align}
The latter in turn leads to transversal $\overline{Z}$, $\overline{S}$, and most importantly a transversal non-Clifford $\overline{T}$ gate. For logical $\overline{Z}$, $\overline{S}$, $\overline{T}$, the value of $\theta$ in Eq.~\eqref{eq:logical} are $\pi$, $\frac{\pi}{2}$ and $\frac{\pi}{4}$ respectively.

To complete the universal logical gate set, we need $\overline{H}$ and $\overline{CNOT}$ or $\overline{CZ}$ gates. Although these gates do not have a simple transversal structure for our $[5,1]$ code, one can use gate teleportation~\cite{Knill_2005_gate_teleportation} to avoid entangling operations between the physical qubits of a given logical block. We refer to Appendix~\ref{appendix:trans_log_gates} for details of the logical Hadamard and CZ gate implementations.

\section{The PAQEC framework}\label{sec:PAQEC}
Having demonstrated the existence of a $5$-qubit code for GAD noise that can correct up to first order in both $\gamma$ and $p$, we now provide an algebraic framework that extends the standard QEC conditions to accommodate such codes. As explained in Sec.~\ref{sec:preli}, the KL conditions for perfect QEC, can be generalized in two ways -- adding a perturbation term leads to approximate QEC (AQEC), whereas, allowing for a trace non-increasing recovery maps leads to probabilistic QEC (PQEC). Here, we combine these two ideas to arrive at conditions for probabilistic approximate QEC (PAQEC), arguably the most general framework under which QEC is possible.

\subsection{Algebraic conditions for PAQEC codes}\label{sec:analytical}

We start with the probabilistic QEC conditions in Eq. \eqref{eq:QEC_prob_cond}. Adding perturbation terms $\beta^{ab}_{ij}(m)$ in the PQEC conditions gives us algebraic conditions for PAQEC as shown in the following theorem.

\begin{Theorem}\label{lemma:paqec}
Given a noise channel $\mathcal{E}$ characterized by Kraus operators $\{G_{n}^{a}\}$ that are grouped into $\mu$ sets $\mathcal{G}^{a}$ ($a\in[1,\mu]$), with each set containing $\eta_{a}$ errors $\{G_{n}^{a}, \, n=1,2,\ldots, \eta_{a}\}$. Consider a quantum code with codewords $\{\ket{i_{L}}, i=1,\ldots, d\}$ that satisfy,

\begin{eqnarray}\label{eq:PAC_QEC_cond2}
    && \sum_{n=1}^{\eta_a} \bra{i_L} G^{(a)\dagger}_n G^{(b)}_m \ket{j_L} = \chi^a_i \delta_{ab} \delta_{ij} + \beta^{ab}_{ij}(m), \nonumber \\
    && \quad \quad \forall m \in [1, \eta_{a}], \; \;  i,j \in [1,d],  \; \;  a,b \in [1, \mu].
  \end{eqnarray}
  for scalars $\{\chi_{i}^{a} \neq 0\}$ and $\{\beta_{ij}^{ab}\}$. Then there exists a recovery channel $\cR$ that achieves an entanglement fidelity {$F_{\rm ent} \geq 1 - \cO(\gamma^{t+1})$} for the given set of errors and can be implemented with a finite probability of success, if the following condition is satisfied.
\begin{align}
 \sum\limits_{a,b,m} \left(\sum\limits_{ij} \left|\frac{\beta^{ab}_{ij}(m)}{\chi^a_i}\right|^2  - \frac{1}{d}\left|\sum\limits_{i}\frac{\beta^{ab}_{ii}(m)}{\chi^a_i}\right|^{2} \right) \leq \cO(\gamma^{t+1}) .\label{eq:th1}
    \end{align}
\end{Theorem}
Note that we obtain the exact PQEC conditions in Eq.~\eqref{eq:QEC_prob_cond} when $\beta^{ab}_{ij}(m) = 0$ in Eq.~\eqref{eq:PAC_QEC_cond2} for all $i,j,a,b$. Alternatively,  approximate QEC conditions are obtained when $\eta_a = 1,~ \forall a\in [1,\mu]$. When both of these constraints hold simultaneously, the conditions in Eq.~\eqref{eq:PAC_QEC_cond2} reduce to the well-known Knill-Laflamme conditions for perfect QEC.

\begin{proof}
We prove the sufficiency of the conditions in Theorem~\ref{lemma:paqec} by explicitly constructing a set of recovery operators that can achieve the desired fidelity, albeit with post-selection.
To this end, we replace $g^a_i(\gamma,p)$ in Eq. \eqref{eq:reco} with $\frac{\lambda}{\chi^a_i}$, and use the recovery channel $\cR$ with Kraus operators $R_a$ {defined therein}.
Here, $\lambda$ is introduced to make the channel $\cR$ trace non-increasing, and is chosen such that the largest eigenvalue of $\sum_a R_a^{\dagger} R_a$ is one.

We now calculate the entanglement fidelity achieved by a quantum code satisfying the conditions in Theorem~\ref{lemma:paqec} for a given noise channel using the expression in Eq. \eqref{eq:ent_fid_prob_case}. {Defining $f^{ab}_{ij} (m) = \frac{\beta^{ab}_{ij}(m)}{\chi^{a}_{i}}$, we get,}
\begin{align}\label{eq:pf2}
    \begin{split}
        & F_{\text{ent}} \\
        =& \frac{\sum\limits_{a,m,i,j} \frac{|\lambda|^2}{d^2} \left(1 + f^{aa}_{ii} (m) + f^{aa}_{jj} (m)^* + \sum\limits_{b} f^{ab}_{ii} (m) f^{ab}_{jj} (m)^* \right)}{\sum\limits_{a,m,i} \frac{|\lambda|^2}{d} \left(1 + f^{aa}_{ii}(m) + f^{aa}_{ii}(m)^* + \sum\limits_{bj} |f^{ab}_{ij}(m)|^2 \right)} \\
        =& \frac{\sum\limits_{a,m,i,j} \frac{1}{d^2} \left(1 + 2 Re( f^{aa}_{ii} (m)) + \sum\limits_{b} f^{ab}_{ii} (m) f^{ab}_{jj} (m)^* \right)}{\sum\limits_{a,m,i} \frac{1}{d} \left(1 +  2 Re(f^{aa}_{ii} (m)) + \sum\limits_{bx} |f^{ab}_{ij}(m)|^2 \right)} \\
        =& \frac{\sum\limits_a \eta_a + \frac{2}{d} \sum\limits_{a,m,i} Re(f^{aa}_{ii}(m)) + \frac{1}{d^2} \sum\limits_{a,b,m} |\sum\limits_i f^{ab}_{ii} (m)|^2}{\sum\limits_a \eta_a + \frac{2}{d} \sum\limits_{a,m,i} Re(f^{aa}_{ii}(m)) + \frac{1}{d} \sum\limits_{a,b,m,i,j} |f^{ab}_{ij}(m)|^2},
    \end{split}
\end{align}

Now, if we choose $\chi^a_i$ to be the mean value of $\sum_{n=1}^{\eta_a} \bra{i_L} G^{(a)\dagger}_n G^{(a)}_m \ket{i_L}$ over $m$, then $\sum_m \beta^{aa}_{ii}(m) = 0$, implying that the linear term $\sum\limits_m f^{aa}_{ii}(m)$ in both numerator and denominator of Eq. \eqref{eq:pf2} vanish.

{Setting $\sum_{a}\eta_{a} = N$, the total number of errors that satisfy the PAQEC conditions for this code, the fidelity-loss evaluates to},
\begin{align}
    &1-F_{\text{ent}} \nonumber\\
    =& \frac{\sum\limits_{a,b,m} \left[\frac{1}{d}\sum\limits_{i,j} \left|f_{ij}^{ab}(m)\right|^2  - \frac{1}{d^2}|\sum\limits_{i}f_{ii}^{ab}(m)|^{2} \right]}{N  +  \frac{1}{d} \sum\limits_{a,b,m,i,j} |f^{ab}_{ij}(m)|^2} \nonumber\\
    =& \frac{1}{\Delta Nd} \sum\limits_{a,b,m} \left( \sum\limits_{ij} \left|f_{ij}^{ab}(m)\right|^2  - \frac{1}{d}\left|\sum\limits_{i} f_{ii}^{ab}(m)\right|^{2} \right) , \label{eq:fl}
    \end{align}
    where,
    \begin{equation}
    \Delta = \left(1 + \frac{\sum\limits_{a,b,m,i,j}|f_{ij}^{ab}(m)|^{2}}{N d}  \right) \geq 1. \label{eq:delta}
\end{equation}

Using the condition in Eq. \eqref{eq:th1} and Eq.~\eqref{eq:delta} we get the desired bound,
\begin{equation}
1-F_{\text{ent}} \leq \frac{\cO(\gamma^{t+1})}{Nd}. \label{eq:proof_bound}
\end{equation}
\end{proof}

Finally, we derive an expression for the success probability of implementing the recovery channel defined in Eq.~\eqref{eq:reco}. For a code satisfying the condition in Eq.~\eqref{eq:th1} and with $\chi^a_i$ chosen to be the mean $\sum_{n=1}^{\eta_a} \bra{i_L} G^{(a)\dagger}_n G^{(a)}_m \ket{i_L}$ over $m$, the average success probability in Eq.~\eqref{eq:avg_prob} evaluates to,

\begin{align}\label{eq:prob_succ}
    \begin{split}
    \Bar{p}_{success} =& |\lambda|^2  \left(\sum_{a} \eta_a + \sum_{a,b,m,i,j} \left|\frac{\beta^{ab}_{ij}(m)}{\chi^a_i} \right|^2 \right).
    \end{split}
\end{align}
Note that while we could ignore the normalizing term $\lambda$ in the fidelity calculation, we must consider it during the calculation of the probability of success.

\subsection{The $[5,1]$ PI code in the PAQEC setup}

We next show that our proposed $[5,1]$ PI code satisfies {a special case of the PAQEC conditions in Theorem~\ref{lemma:paqec} for single-qubit GAD noise.}
\begin{corollary}\label{lemma:1}
     The $[5,1]$ PI code given in Eq.~\eqref{eq:code_dutta5} satisfies the conditions  in Eq. \eqref{eq:PAC_QEC_cond2} and Eq.~\eqref{eq:th1} for single-qubit generalized amplitude-damping errors with $t=1$, hence achieving an entanglement fidelity $F_{\rm ent} \geq 1 - \mathcal{O}(\gamma^2)$.
\end{corollary}
\begin{proof}
{As already noted in Sec.~\ref{subsec:encoding}, the action of the first order GAD errors belonging to the different groups $\{\cG^{(0)}, \cG^{(\uparrow)}, \cG^{(\downarrow)}\}$ maps the codewords in Eq.~\eqref{eq:code_dutta5}  into orthogonal states. Errors from the same group map codewords $\ket{0_L}$ and $\ket{1_L}$ into a pair of orthogonal states. Hence, there is no overlap between the states $G^{(a)}_v \ket{i_L}$ and $G^{(b)}_m \ket{j_L}$, when $a\neq b$ or $i\neq j$. This implies that a special case of the PAQEC condition in Eq.~\eqref{eq:PAC_QEC_cond2} is satisfied, with,}
\begin{equation}\label{eq:D1}
   \beta^{ab}_{ij}(m) \propto \delta_{ab} \delta_{ij}.
\end{equation}

Now, to prove the Corollary \ref{lemma:1}, we numerically check that
\begin{align}\label{eq:ineq1}
 \left( \sum\limits_{i} \left|\frac{\beta^{aa}_{ii}(m)}{\chi^a_i}\right|^2  - \frac{1}{d}\left|\sum\limits_{i}\frac{\beta^{aa}_{ii}(m)}{\chi^a_i}\right|^{2} \right) \leq \alpha_{a}(p)\cO(\gamma^{2})
\end{align}
for all $a \in \{0, \uparrow, \downarrow\}, m \in [1, \eta_a]$. Hence, we can conclude that the inequality in Eq. \eqref{eq:th1} holds true with $t=1$, as desired.

\end{proof}

\subsection{Choice of Recovery Channel for PAQEC}\label{sec:diff_reco}

There is no fixed way to choose the value of $\chi^a_i$ for PAQEC codes, as there are many ways to separate the $m$-independent $\chi^a_i$ and $m$-dependent $\beta^{aa}_{ii}(m)$ in Eq. \eqref{eq:PAC_QEC_cond2}.
We can therefore obtain different recovery channels corresponding to different choices of ${\chi^a_i}$. We will describe three such canonical choices below and compare the corresponding entanglement fidelities obtained for the $[5,1]$ PI code with GAD noise.

\textit{\bf $\bullet$ Average value --} The default choice we have made in Sec.~\ref{subsec:recovery} is to let $\chi^a_i$ be the average value of $\sum_{v} \bra{i_L} G^{(a)\dagger}_v G^{(a)}_m\ket{i_L}$. {The corresponding recovery channel for GAD noise with the $[5,1]$ PI code successfully corrects all the first-order errors, since there is no linear $\gamma$ term present in the entanglement fidelity evaluated in Eq. \eqref{eq:ent_fid_avg}.}
However, there are other choices of $\{\chi^a_i\}$ that can give better entanglement fidelity.

\begin{figure*}[t!]
    \includegraphics[width=0.95\linewidth]{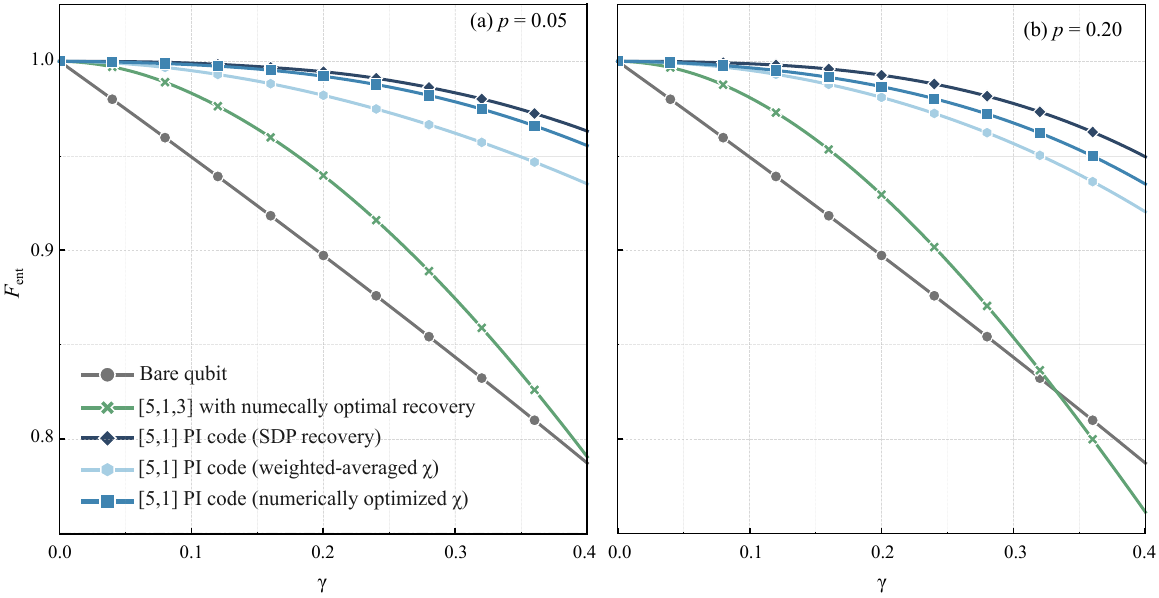}
    \caption{
    \textbf{Performance of different noise-adapted recoveries, for the $[5,1]$ PI code. }
    Entanglement fidelity versus the damping strength $\gamma$ for (a) $p = 0.05$ and (b) $p = 0.2$, for the $[5,1]$ permutation-invariant code with different noise-adapted recoveries,  the $[5,1,3]$ code (numerically optimized recovery), and the bare qubit.
    }
    \label{fig:ent_fid_comp}
\end{figure*}

\textit{\bf $\bullet$ Largest value --} One can choose $\chi^a_i$ to be the largest possible value of $\sum_{v}\bra{i_L} G^{(a)\dagger}_v G^{(a)}_m\ket{i_L}$, by maximizing over all $m$.
\[\chi^a_i = \max\limits_{m}{\sum_v \bra{i_L} G^{(a)\dagger}_v G^{(a)}_m\ket{i_L}}.\]
{In this case, the resulting recovery channel is primarily adapted to correct for} the error occurring with the highest probability. {For the $[5,1]$ PI code with GAD noise,} this yields an entanglement fidelity
\begin{align}\label{eq:highest_ent_fid}
    F_{\text{ent}} = 1 - \left(\frac{9}{20} p + \frac{9}{5} p^2 \right) \gamma^2 + \cO(\gamma^{3})
\end{align}
Choosing the largest value over the average value of $\sum_v \bra{i_L} G^{(a)\dagger}_v G^{(a)}_m\ket{i_L}$ thus gives a higher entanglement fidelity as seen by comparing Eqs. \eqref{eq:ent_fid_avg} and~\eqref{eq:highest_ent_fid}.

\textit{\bf $\bullet$ Weighted average value --}
We can improve the entanglement fidelity further by choosing the value of $\chi^a_i$ to be a \emph{weighted average} of $\sum_{v}\bra{i_L} G^{(a)\dagger}_v G^{(a)}_m\ket{i_L}$ over $m$. To compute the weights for the GAD channel with noise parameters $p$ and $\gamma$, we first count the number of single-qubit damping errors and the number of single-qubit excitation errors present in the tensor product decomposition of each error operator of the channel. If $x_m$ is the number of damping errors and $(5-x_m)$ is the number of excitation errors, then the probability or weight of that error is proportional to $\left((1-p)^{x_m} p^{5-x_m}\right)$. For example, {the weight associated with the error $A_0\otimes R_1 \otimes A_0 \otimes R_0 \otimes A_0$ is $(1-p)^3p^2$}. The weighted average choice leads to,
\begin{align}\label{eqn:wt_avg_expression}
    \chi_i^{a} = \frac{\sum_{m} (1-p)^{x_m} p^{5-x_m} \sum_{v}\bra{i_L} G^{(a)\dagger}_v G^{(a)}_m\ket{i_L}}{\sum_{m} (1-p)^{x_m} p^{5-x_m}}
\end{align}
{The corresponding entanglement fidelity for the $[5,1]$ PI code is,}
\begin{align}
    {F}_{\text{ent}} =& 1 - \left( \frac{9}{20} p + \frac{9}{5} p^2 - \frac{9}{2} p^{5/2} + \cO(p^3) \right) \gamma^2 + \cO(\gamma^{3}).
\end{align}
{This fidelity is marginally higher than that achieved by the previous recoveries, implying} that the best analytical choice for $\chi^a_i$ is indeed the weighted-average approach. Nevertheless, this choice does not guarantee the optimal choice of ${\chi^a_i}$ to maximize entanglement fidelity.

\textit{\bf $\bullet$ Numerically optimized --} Finally, to obtain the optimal choice for $\{\chi^a_i\}$, we can numerically optimize the fidelity using the form of the recovery operators given in Eq. \eqref{eq:reco_PQEC_perfect}.
Algorithm~\ref{alg:find-loss} below outlines the procedure used to compute {fidelity loss, which is then minimized using Powell's Method~\cite{Virtanen2020}}. The fidelity obtained with the optimized set $\{\chi^a_i\}$ exceeds that of previous analytical choices, as shown in Fig.~\ref{fig:ent_fid_comp}.
\begin{algorithm}[H]
\caption{Calculating Loss function to optimize $\{\chi^a_i\}$}
\label{alg:find-loss}
\begin{algorithmic}[1]
\REQUIRE $\{\chi'^a_i\}, \{ G^{(a)}_m \}$
\STATE ${R}_a \gets \sum_{im} \frac{1}{{\chi'}_i^a} \dyad{i_L} G^{(a)\dagger}_m$
\STATE $\text{Loss} \gets 1 - F_{\text{ent}}\left( \{\ket{i_L}\}, \{ G^{(a)}_m \}, \{ R_a \} \right)$
\RETURN Loss
\end{algorithmic}
\end{algorithm}

\section{PAQEC as an optimization problem}\label{sec:Optimization}
Having established PAQEC as a viable framework for QEC, we now explore the possibility of going beyond analytical constructions and seek better performance through numerical optimization. 

\subsection{Optimizing PAQEC recovery using semi-definite programming}
The semi-definite programming (SDP) framework~\cite{Berta2021, Gerard_2024} provides a numerically efficient way to solve convex optimization problems for which the objective function is linear, with the input constrained to a semidefinite cone. In past work, the problem of finding good \emph{approximate} QEC (AQEC) codes has been recast as an SDP, and numerically optimized recovery channels have been constructed within the framework of noise-adapted AQEC codes~\cite{fletcher_thesis, Fletcher_sdp}.

Here, we first show that the problem of PAQEC can be formulated as a convex optimization task. We then use SDP to find the optimal recovery channel for PAQEC, without imposing any specific structural constraints on the recovery operation. Note that the different noise-adapted recovery maps considered in Sec.~\ref{sec:diff_reco} have a rather specific structure, as given in Eq.~\eqref{eq:reco_PQEC_perfect}.

Consider a map $\cE$ composed of a unitary encoding circuit and a noise channel such that $\cE: \cL(\cH_S) \to \cL(\cH_C)$, where $\cH_S$ and $\cH_C$ are the Hilbert spaces of the unencoded and encoded qubits, respectively.
Similarly, consider a recovery map $\cR$ that captures the joint action of the recovery channel and the decoding unitary, $\cR: \cL(\cH_C) \to \cL(\cH_S)$. In our case, the recovery channel is completely positive but \emph{trace non-increasing}.
Here, $\cL(\cH)$ denotes the set of linear operators acting on the Hilbert space $\cH$.

We begin by rewriting the expression for the entanglement fidelity in Eq.~\eqref{eq:ent_fid_prob_case} for probabilistic QEC as,
\begin{align}
    F_{\text{ent}} &= \frac{1}{\Bar{p}_{\text{success}}} \llangle \rho | X_{\cR \circ \cE} | \rho \rrangle, \label{eq:ent_fid_SDP1}
\end{align}
where, $\rho$ is the maximally mixed state and $\Bar{p}_{\text{success}}$ is the average success probability defined in Eq.~\eqref{eq:avg_prob}. 
$|A\rrangle \equiv \operatorname{vec}(A)$ denotes the vectorization of an operator $A$, with  $\llangle A| = \operatorname{vec}(A)^\dagger$~\cite{fletcher_thesis}. $X_{\cR \circ \cE}$ is the Choi matrix obtained by the action of the combined channel $\cR \circ \cE$ on the maximally entangled state with the following vectorized form,
\begin{align}
    X_{\cR \circ \cE} &= \sum_{ij} |R_i E_j\rrangle \llangle R_i E_j| \nonumber\\
    &= \sum_{ij} (E^{T}_j \otimes I) |R_i\rrangle \llangle R_i | (E^{T}_j \otimes I)^{\dagger} \nonumber\\
    &= \sum_{j} (E^{T}_j \otimes I) X_{\cR} (E^{T}_j \otimes I)^{\dagger}.
\end{align}
$X_{\cR}$ is the Choi matrix of the recovery channel. Finally, we express the average probability of successfully implementing the recovery as,
\begin{align}
    \Bar{p}_{\text{success}} &= \sum_{ij}\Tr[(I \otimes R_i E_j) \dyad{\Psi} (I \otimes R_i E_j)^{\dagger}] \nonumber\\
    &= 2 \sum_{ij} \Tr[(I \otimes R_i E_j) |\rho\rrangle \llangle \rho| (I \otimes R_i E_j)^{\dagger}] \nonumber\\
    &= 2 \sum_{ij} \Tr[|R_i E_j\rho\rrangle \llangle R_i E_j\rho|] \nonumber\\
    &= 2 \Tr[X_{\cR} B_{\cE,\rho}] \label{eq:succ_prob_sdp}
\end{align}
where, $\ket{\Psi}$ is the purification of the maximally mixed state which is equivalent to $\sqrt{2}|\rho\rrangle$ in the vectorized notation, and $B_{\cE,\rho} := \sum_{j} (\rho^T E_j^T \otimes I)^{\dagger}(\rho^T E_j^T \otimes I)$.

Combining Eqs.~\eqref{eq:ent_fid_SDP1} and~\eqref{eq:succ_prob_sdp}, we can express the entanglement fidelity for a PAQEC protocol as,
\begin{align}\label{eq:ent_fid_SDP}
    F_{\text{ent}} = \frac{\Tr[X_{\cR} C_{\cE, \rho}]}{2\Tr[X_{\cR}B_{\cE, \rho}]},
    \end{align}
where,
\begin{equation}
    C_{\cE, \rho} = \sum_{j} (E^{T}_j \otimes I)^{\dagger} |\rho \rrangle \llangle \rho | (E^{T}_j \otimes I).
\end{equation}
As our recovery channel is completely positive (CP) and trace non-increasing, we get the following positive semi-definite constraint equations.
\begin{align}
    &X_{\cR} \succeq 0, ~~ \text{(CP condition)} \\
    & \Tr_{\cH_S}\{X_{\cR}\} \preceq I_{\cH_C}~~ \text{(Trace Non-Increasing)}
\end{align}
In this form, the objective is a ratio of two quantities linear in $X_{\cR}$ and hence not itself linear, so Eq.~\eqref{eq:ent_fid_SDP} is a linear-fractional program rather than an SDP. We will use the Charnes--Cooper transformation~\cite{Charnes1962} to rewrite it into a linear form. We define,
\begin{align}
    Y = \alpha X_{\cR},~\text{where } \alpha = \frac{1}{\Tr[X_{\cR}B_{\cE, \rho}]}\label{eq:alpha_def}
\end{align}
Then, the expression for entanglement fidelity becomes,
\begin{equation}
    F_{\text{ent}} = \frac{1}{2}\Tr[Y C_{\cE, \rho}]. \label{eq:ent_fid_SDP_final}
\end{equation}
To obtain the optimal value of $F_{\text{ent}}$, we need to maximise the objective function in Eq.~\eqref{eq:ent_fid_SDP_final} over $Y$, subject to the following constraint equations:
\begin{align}
    & Y \succeq 0, \\
    & \Tr[YB_{\cE, \rho}] = 1, \\
    & \alpha \geq 0, \\
    & \Tr_{\cH_S}\{Y\} \preceq \alpha I_{\cH_C}.
\end{align}
Once we have determined the optimal fidelity, we will then proceed to find $X_{\cR}$ such that $\Tr[X_{\cR}B_{\cE, \rho}]$ is maximum, to maximize the probability of successfully implementing the optimal recovery. 

Algorithm~\ref{alg:optimal-probabilistic} presents the full SDP-based algorithm for finding the optimal recovery for a given probabilistic QEC code and noise model.
\begin{algorithm}[H]
\caption{Optimal probabilistic recovery via Charnes--Cooper transformation}
\label{alg:optimal-probabilistic}
\begin{algorithmic}[1]
\REQUIRE $C_{\cE,\rho},\, B_{\cE,\rho}$ [Noise,Code]
\ENSURE Choi matrix $X_{\cR}$ of the optimal recovery and its success probability $\Bar{p}_{\text{success}}$

\STATE Define optimization variables:
\[
Y \succeq 0,\qquad \alpha \ge 0.
\]

\STATE Solve the semidefinite program
\[
\begin{aligned}
\max_{Y,\alpha}\quad &
\frac{1}{2}\Tr[Y C_{\cE,\rho}] \\
\text{subject to}\quad
& \Tr[YB_{\cE,\rho}]=1,\\
& \Tr_{\cH_S}\{Y\} \preceq \alpha I_{\cH_C},\\
& \alpha \geq 0,\\
& Y\succeq 0.
\end{aligned}
\]

\STATE Recover the Choi matrix
\[
\alpha \gets \lambda_{\max}\big(\Tr_{\cH_S}\{Y\}\big), \qquad
X \gets \frac{Y}{\alpha}.
\]

\STATE Compute the optimal fidelity
\[
F \gets
\frac{\Tr[X C_{\cE, \rho}]}{2\Tr[X B_{\cE, \rho}]}.
\]

\STATE $X_{\cR} \gets \arg\max_{X \succeq 0,\; \text{Tr-NI}} \Re\big[\Tr(X\, B_{\cE,\rho})\big]$ s.t. \\
$\Re\big[\Tr(X\, C_{\cE,\rho} - 2FX\, B_{\cE,\rho})\big] = 0$

\RETURN $X_{\cR}$$,\ \Bar{p}_{\text{success}} = 2\Tr[X_{\cR}B_{\cE,\rho}]$
\end{algorithmic}
\end{algorithm}

The  optimization in Step 2 of Algorithm~\ref{alg:optimal-probabilistic} leaves $\alpha$ undetermined, since it enters only through $\Tr_{\cH_S}\{Y\} \preceq \alpha I_{\cH_C}$. This is the trace non-increasing condition $\sum_i R_i^{\dagger}R_i = \Tr_{\cH_S}\{Y\}/\alpha \preceq I_{\cH_C}$ written for $X_{\cR} = Y/\alpha$. Increasing $\alpha$ reduces $\sum_i R_i^{\dagger}R_i$, which in turn reduces the success probability. Changing $\alpha$ does not affect $F_{\text{ent}}$, since the expression in Eq.~\eqref{eq:ent_fid_SDP} is invariant under $X_{\cR} \mapsto sX_{\cR}$. Step 3, therefore, takes $\alpha = \lambda_{\max}(\Tr_{\cH_S}\{Y\})$, which is the smallest admissible value for which the recovery is successfully implemented as often as the complete positivity and trace non-increasing constraints allow. From Eqs.~\eqref{eq:succ_prob_sdp} and \eqref{eq:alpha_def} we see that the largest success probability compatible with the optimal fidelity is therefore, $\Bar{p}_{\text{success}} = 2/\alpha$.

Fig.~\ref{fig:ent_fid_comp} shows that the SDP-based optimal recovery achieves the highest entanglement fidelity, followed by the recovery with numerically optimized $\chi^{(a)}_i$ and weighted averaged $\chi^{(a)}_i$.
However, a higher entanglement fidelity may not translate into a higher success probability.
Hence, in practice, it is sometimes a better idea to optimize for the probability of success with a slightly reduced entanglement fidelity.

\subsection{Optimal fidelity at fixed success probability}\label{sec:frontier}

The trade-off between the success rate and the entanglement fidelity is not specific to any particular recovery; it simply follows from the structure of the objective function in defined in Eq.~\eqref{eq:ent_fid_SDP1}. We formally characterize this tradeoff in Lemma~\ref{prop:frontier} below. Writing $\Bar{p}$ for $\Bar{p}_{\text{success}}$, we define for fixed $\Bar{p}$, 
\begin{align}\label{eq:Nc}
    N(\Bar{p}) := \max_{X_{\cR}} \Big\{ \Tr[X_{\cR} C_{\cE,\rho}] \;:\; & 2\Tr[X_{\cR} B_{\cE,\rho}] = \Bar{p}, \nonumber \\
    & X_{\cR} \succeq 0, \; \Tr_{\cH_S}\{X_{\cR}\} \preceq I_{\cH_C} \Big\}.
\end{align}
Then, $F^{*}(\Bar{p}) = N(\Bar{p})/\Bar{p}$ is the largest entanglement fidelity attainable at success probability $\Bar{p}$. Fixing $\Bar{p}$ pins the denominator of Eq.~\eqref{eq:ent_fid_SDP}, so Eq.~\eqref{eq:Nc} is an ordinary semidefinite program. 

\begin{figure*}
    \centering
    \includegraphics[width=0.95\linewidth]{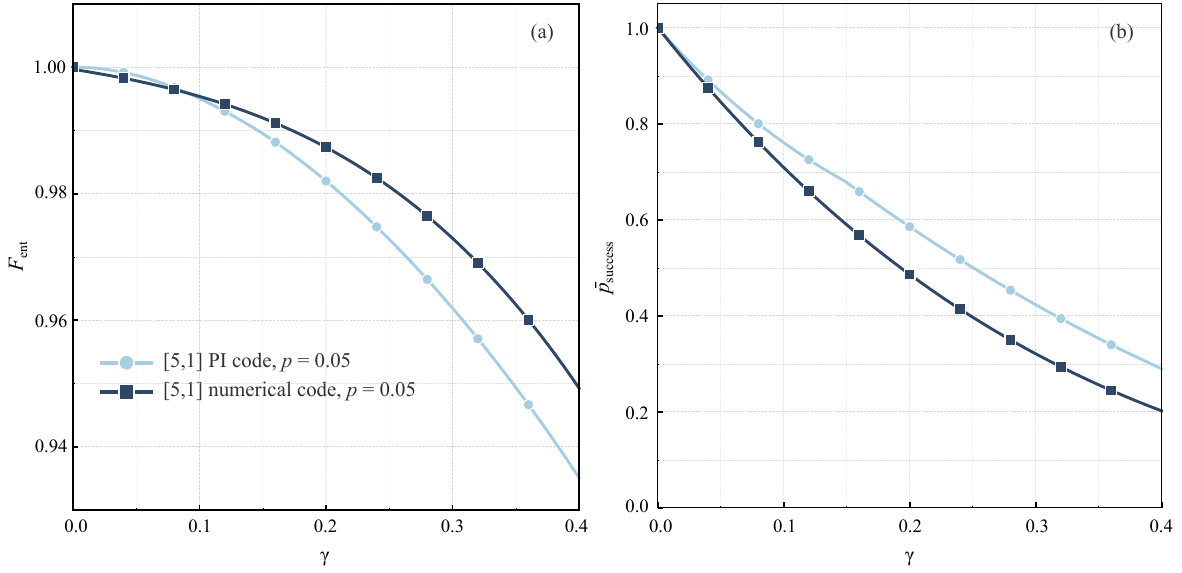}
     \caption{\textbf{Performance comparison of the $[5,1]$ PI code and the numerically optimized code.}
(a) Entanglement fidelity vs.~$\gamma$. (b) Recovery probability vs.~$\gamma$. The numerical code outperforms the PI code at $p=0.05$. This improvement comes at the cost of a lower probability of success (increased overhead).}
    \label{fig:manav_vs_dutta}
\end{figure*}

\begin{lemma}\label{prop:frontier}
~
\begin{itemize}
    \item $F^{*}$ is non-increasing on $(0,1]$.
    \item Let $F$ be the optimal fidelity computed in Step 4 of Algorithm~\ref{alg:optimal-probabilistic}. Then there exists $\Bar{p}_{\max} \in (0,1]$ such that $N$ is linear, $N(\Bar{p}) = F\Bar{p}$, on $(0,\Bar{p}_{\max}]$. 
    \item For every $\Bar{p} > \Bar{p}_{\max}$:
$$F^{*}(\Bar{p}) < F.$$
\end{itemize}

\end{lemma}

\begin{proof}
If $X_{\cR}$ is feasible in Eq.~\eqref{eq:Nc} at $\Bar{p}$ and $0 < s \leq 1$, then $sX_{\cR}$ is feasible at $s\Bar{p}$, since scaling preserves positivity and only relaxes $\Tr_{\cH_S}\{X_{\cR}\} \preceq I_{\cH_C}$. Taking $X_{\cR}$ to be the optimal recovery at $\Bar{p}$, so that $\Tr[X_{\cR}C_{\cE,\rho}] = N(\Bar{p})$, the objective of Eq.~\eqref{eq:Nc} evaluated at $sX_{\cR}$ is $\Tr[sX_{\cR}C_{\cE,\rho}] = sN(\Bar{p})$. 
Since $N(s\Bar{p})$ is the maximum over all points feasible there,
\begin{align}
N(s\Bar{p}) &\geq sN(\Bar{p})\\
\implies F^{*}(s\Bar{p}) = \frac{N(s\Bar{p})}{(s\Bar{p})} &\geq \frac{N(\Bar{p})}{\Bar{p}} = F^{*}(\Bar{p}).
\end{align}
Every pair $0 < \Bar{p}_1 < \Bar{p}_2$ is of this form, with $s = \Bar{p}_1/\Bar{p}_2$, so $F^{*}$ is non-increasing.

Since $F$ is the largest attainable fidelity computed by Algorithm~\ref{alg:optimal-probabilistic}, $N(\Bar{p}) \leq F\Bar{p}$ throughout. 

Let $\Bar{p}_{\max}$ be the largest success probability at which $F$ is attained. Scaling the corresponding optimal recovery as above gives $N(\Bar{p}) = F\Bar{p}$ for every $\Bar{p} \leq \Bar{p}_{\max}$. Beyond $\Bar{p}_{\max}$, no recovery attains $F$, so $F^{*}(\Bar{p}) < F$. 
\end{proof}

The threshold for the probability of successfully implementing the PAQEC recovery is
\begin{align}\label{eq:pmax}
    \Bar{p}_{\max} = \max_{Y}\ \frac{2}{\lambda_{\max}(\Tr_{\cH_S}\{Y\})},
\end{align}
with the maximum taken over the optimizers $Y$ of Step 2 in Algorithm~\ref{alg:optimal-probabilistic}. Steps 3 and 5 evaluate Eq.~\eqref{eq:pmax} between them, the former for the optimizer that Step 2 returns and the latter over all of them. Algorithm~\ref{alg:optimal-probabilistic}, therefore, returns a recovery that yields both the maximum attainable entanglement fidelity as well as the highest probability of success, for a given scenario. Taking any probability of success beyond $\Bar{p}_{\max}$ will reduce the $F$ achieved by the recovery.

\subsection{Finding good PAQEC codes}\label{sec:good_PAQECC}
Finally, using the PAQEC conditions, we may also develop a numerical approach to find better codes. We first rewrite the conditions in Eq.~\eqref{eq:PAC_QEC_cond2} as,
\begin{align}\label{eqn:pac_qec_2}
    \sum_{v=1}^{\eta_a} \bra{i} T^\dagger G^{(a)\dagger}_v G^{(b)}_m T\ket{j}
    =
    \chi^a_i \delta_{ab} \delta_{ij} + \beta^{ab}_{ij}(m),
\end{align}
\noindent where $T$ is the encoding isometry used to transform all states from physical to logical space, $T|i\rangle = |i_L\rangle$. Naturally, $T$ will be rectangular of the form $2^n\times 2^k$ for an $[n, k]$ code. We now formulate an optimization problem to find the best encoding matrix $T$ with $2^{n+k}$ real entries, for a given noise channel with a grouping of errors into different sets, as discussed in Sec.~\ref{sec:PAQEC}. We refer to Appendix \ref{sec:App_Optimization} for a detailed discussion of our numerical search algorithm for good PAQEC codes.

Fig. \ref{fig:manav_vs_dutta} compares the entanglement fidelity and the probability of success of the $[5,1]$ PI code and the numerically optimized code for GAD noise.
The numerical code was obtained by sampling error operators with noise parameters $\gamma \in (0,0.4]$ and $p \in (0, 0.2]$.
As seen in Fig. \ref{fig:manav_vs_dutta}, the numerical code outperforms the PI code for $p = 0.05$, indicating that our search algorithm can yield codes superior to the analytically derived PI code in regimes where the relevant noise parameters fall below the sampling range.

\section{Conclusions and future directions}\label{sec:conclusion}

In this work, we have constructed a five-qubit quantum code for generalized amplitude-damping noise, which forms an important source of errors in many quantum computing platforms. Departing from the conventional deterministic recovery paradigm, we construct a probabilistic recovery scheme using which our five-qubit permutation-invariant code achieves a fidelity loss that scales quadratically with noise strength. 

Our probabilistic QEC scheme for GAD noise outperforms the known, deterministic QEC codes, including the $[5,1,3]$ stabilizer code, all of which exhibit {linear fidelity} loss. We present a syndrome-based recovery protocol for implementing our probabilistic QEC scheme and also construct a universal set of transversal logical gates for the $5$-qubit PI code.

Motivated by the structure of our $5$-qubit code, we develop an algebraic framework for probabilistic approximate quantum error correction that relaxes the original Knill-Laflamme conditions in the most general possible way. We further show that the PAQEC conditions can be formulated as an SDP-based optimization problem to find the optimal recovery channel that maximizes the entanglement fidelity. This, in turn, allows us to formalize the trade-off between entanglement fidelity and the probability of successful recovery. 
Finally, we present a systematic optimization-based approach to discovering good quantum codes under the PAQEC framework.

Going forward, it may be interesting to use the PAQEC approach to correct for other realistic noise sources in quantum hardware. 
This direction holds promise for unlocking new classes of efficient, practical, and hardware-tailored quantum error correcting codes. 
Another promising direction of study would be to identify low-depth circuit implementations of the $[5,1]$ PI encoding and recovery, which could in turn pave the way for an experimental realization, even on near-term quantum devices. 
An immediate question for future work is whether the $[5,1]$ PI code can be generalized to a broader family of codes tailored to generalized amplitude damping (GAD) noise.
Finally, the existence of a complete set of transversal logical gates for the $[5,1]$ PI code also suggests a potential pathway to construct fault-tolerant gadgets and estimate the corresponding thresholds for GAD noise.

\section*{Acknowledgements}
We acknowledge the grant support provided by the Mphasis F1 Foundation to CQuICC, a portion of which contributed to the completion of this project. SD, AR, and PM acknowledge funding from the National Quantum Mission, Department of Science and Technology, Govt. of India, via grant no. DST/QTC/NQM/QC/2024/1.
\bibliography{ref}

\appendix

\section{Probability of Success}\label{app:A}

Consider a logical qubit with codewords $\{\ket{0_L}, \ket{1_L}\}$, and an arbitrary logical state
$$
\ket{\psi_L} = \cos\frac{\theta}{2}\,\ket{0_L} + e^{i\phi}\sin\frac{\theta}{2}\,\ket{1_L},
$$
where $(\theta,\phi)$ parameterize points on the Bloch sphere.
Let the probabilistic recovery channel be represented by the Kraus operators $\{R_m\}$, which are applied with post-selection. The success probability for an input logical state $\rho_L$ is then given by
\begin{equation}
\mathbb{P}(\rho_L) = \mathrm{Tr}\!\left[\sum_m R_m \cE(\rho_L) R_m^\dagger\right]
= \mathrm{Tr}\!\left(\rho_L \Phi\right),
\end{equation}

\begin{figure*}[t]
    \centering
    \includegraphics[width=0.9\linewidth]{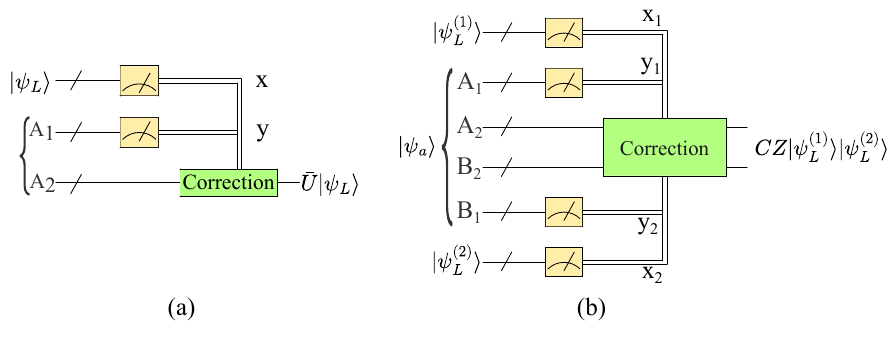}
    \caption{\textbf{Transversal logical gates using gate teleportation.} (a) Logical gate $\bar{U}$, with the ancilla qubits $A_1$ and $A_2$ prepared in $\frac{1}{\sqrt{2}}\left(|0_L\rangle \otimes\Bar{U} |0_L\rangle + |1_L\rangle \otimes \Bar{U}|1_L\rangle \right)$, (b) Two-qubit $\overline{CZ}$ gate for the $[5,1]$ PI code.}
    \label{fig:gate}
\end{figure*}

where $\cE$ is the noise channel with Kraus operators $\{E_j\}$ and
$$\Phi := \sum_{jm} E_j^\dagger R_m^{\dagger} R_m E_j$$
is the POVM element associated with the successful outcome ($0 \le \Phi \le I$).
For a pure logical state $\rho_L = \ket{\psi_L}\!\bra{\psi_L}$, the success probability is
$$
\mathbb{P}(\psi_L) = \bra{\psi_L}\Phi\ket{\psi_L}.
$$
The average success probability, taken uniformly over all pure logical states on the Bloch sphere (i.e., with respect to the Haar measure $d\psi$), is
\begin{equation}
\overline{p}_{\text{success}} = \int \bra{\psi_L}\Phi\ket{\psi_L}\, d\psi_L
= \mathrm{Tr}\!\left(\Phi \int \ket{\psi_L}\!\bra{\psi_L}\, d\psi_L\right).
\end{equation}
Using the well-known identity for the first-order Haar average over an $\Bar{d}$-dimensional Hilbert space,
$$
\int \ket{\psi}\!\bra{\psi}\, d\psi = \frac{I}{\Bar{d}},
$$
we obtain, for the logical qubit ($\Bar{d}=2$),
\begin{equation}
\overline{p}_{\text{success}}
= \mathrm{Tr}\!\left(\Phi\,\rho_L\right)
= \mathbb{P}\!\left(\frac{I}{2}\right).
\end{equation}
\noindent
Hence, the Haar-averaged probability of success equals the success probability evaluated on the maximally mixed logical state.

\section{Transversal Logical Gates for the $[5,1]$ PI Code}\label{appendix:trans_log_gates}

Sec.~\ref{sec:logical_gates} already describes a set of single-qubit logical gates for the $[5,1]$ PI code, namely, the logical $\bar{Z}$, $\bar{S}$ and $\bar{T}$. We now describe circuits for the logical Hadamard and Control-$Z$ gates using gate teleportation~\cite{Knill_2005_gate_teleportation}.

To implement the logical Hadamard gate, we use a resource state of the form $\frac{1}{\sqrt{2}}(\ket{0_L}\otimes\ket{+_L}+\ket{1_L}\otimes\ket{-_L})$.
As shown in Fig. \ref{fig:gate} (a), we measure the logical qubits and the first block of ancilla qubits in the logical Bell basis. Depending on the outcome, we apply a suitable correction (either logical $\overline{X}$ or logical $\overline{Z}$ or both) to apply the logical $\overline{H}$ gate.
To implement the logical $\overline{CZ}$ gate using gate teleportation, we prepare the ancilla qubits in the state \begin{align}
    \ket{\psi_{a}} &=\frac{1}{\sqrt{2}} \ket{0_L}_{A_1}\ket{0_L}_{A_2}\left(\frac{\ket{0_L}_{B_1}\ket{0_L}_{B_2}+\ket{1_L}_{B_1}\ket{1_L}_{B_2}}{\sqrt{2}}\right) \nonumber\\
    &+ \frac{1}{\sqrt{2}} \ket{1_L}_{A_1}\ket{1_L}_{A_2}\left(\frac{\ket{0_L}_{B_1}\ket{0_L}_{B_2}-\ket{1_L}_{B_1}\ket{1_L}_{B_2}}{\sqrt{2}}\right).
\end{align}

Consider a pair of logical qubits in states $|\psi_L^{(1)}\rangle$ and $|\psi_L^{(2)}\rangle$ respectively. To implement a logical $\overline{CZ}$ across them, we first measure $\{|\psi_L\rangle, A_1\}$ and $\{|\psi_L^{(2)\rangle},B_1\}$ in the logical Bell basis $\{\ket{\phi^{xy}} = \frac{1}{\sqrt{2}}(\ket{0_L}\ket{y_L} + (-1)^x \ket{1_L}\ket{\Bar{y}_L})\}$, as shown in Fig \ref{fig:gate} (b).
Based on the measurement outcome $\{x,y\}$, one can apply suitable $\overline{X}$ and $\overline{Z}$ corrections as shown in Table \ref{tab:Correction} to implement $\overline{CZ}$ across the two logical blocks.

\begin{table}[h!]
    \centering
\begin{tabular}{|c|c|c||c|c|c|}
\hline
$\{x_1,y_1\}$ & $\{x_2,y_2\}$ & Correction & $\{x_1,y_1\}$ & $\{x_2,y_2\}$ & Correction\\
\hline
$\{0,0\}$ & $\{0,0\}$ & $I \otimes I$ &$\{0,1\}$ & $\{0,0\}$ & $X \otimes Z$\\
$\{0,0\}$ & $\{0,1\}$ & $Z \otimes X $&$\{0,1\}$ & $\{0,1\}$ & $X \otimes ZX $ \\
$\{0,0\}$ & $\{1,0\}$ & $I \otimes Z $ &$\{0,1\}$ & $\{1,0\}$ & $X \otimes I$\\
$\{0,0\}$ & $\{1,1\}$ & $Z \otimes ZX$&$\{0,1\}$ & $\{1,1\}$ & $ ZX \otimes X $ \\
$\{1,0\}$ & $\{0,0\}$ & $Z \otimes I$&$\{1,1\}$ & $\{0,0\}$ & ${ZX} \otimes {Z}$ \\
$\{1,0\}$ & $\{0,1\}$ & $I \otimes X$ &$\{1,1\}$ & $\{0,1\}$ & $X \otimes ZX$\\
$\{1,0\}$ & $\{1,0\}$ & $Z \otimes Z$ &$\{1,1\}$ & $\{1,0\}$ & $ZX \otimes I$\\
$\{1,0\}$ & $\{1,1\}$ & $I \otimes ZX$&$\{1,1\}$ & $\{1,1\}$ & $X \otimes X$ \\
\hline
\end{tabular}
\caption{Correction needed after the measurement outcome $\{x_1,y_1\}$ for $\ket{\psi_L^{(1)}}$ and $A_1$ qubit of $\ket{\psi_a}$ and $\{x_2,y_2\}$ for $\ket{\psi_L^{(2)}}$ and $B_1$ qubit of $\ket{\psi_a}$.}
\label{tab:Correction}
\end{table}

\section{Numerical Search for PAQEC Codes} \label{sec:App_Optimization}

We describe here the complete numerical procedure used to obtain optimal PAQEC codes, starting with our algebraic conditions for PAQEC, as outlined in Sec.~\ref{sec:good_PAQECC}. We start by listing all the different pieces of the loss function that we will optimize over, to find the desired codewords.
\begin{itemize}\setlength{\itemsep}{0.1cm}
    \item \textbf{$L_{\text{code-ortho}}: $} $\sum_{i\neq j} |\langle i_L | j_L \rangle| \rightarrow 0$, for mutual orthogonality of codewords.
    \item \textbf{$L_{\text{code-norm}}: $} $\sum_i (1 - |\langle i|i \rangle|^2) \rightarrow 0$, to ensure $\text{norm} = 1$ for each code.
    \item \textbf{$L_{\text{-norm}}: $} $\epsilon \cdot \sum_{a, m} \sum_i (1 - ||G^{(a)}_m|i_L\rangle||)^2 \rightarrow 0$, normality of all code words in its noise space with a small weight ($\epsilon$), to prevent a vanishing codeword
    \item \textbf{$L_{\text{PA-ortho}}: $} $\sum_{i, j} \sum_{m, n} ||\langle i|T^\dagger G_m^{(a)\dagger} G_n^{(b)}T|j \rangle||_2\rightarrow 0\ \because\ a\neq b \text{ and } i \neq j \text{ when } a = b \ \forall a,b$, orthogonality between different partitions of error operators via their noise spaces and orthogonality between different code words in a given noise space respectively.
    \item \textbf{$L_{\text{PA-lap}}: $} $\sum_{i} \sum_n ||\sum_m \langle i|T^\dagger G_m^{(a)\dagger} G_n^{(a)}T|i \rangle - \chi_i^a ||_2\rightarrow 0\ \forall\ a$, constant overlap between different codewords for a given partition.
    \item \textbf{$L_{\text{fid}}: $} $(1 - F_\text{ent})^2 \rightarrow 0$, ensure tends to $1$.
\end{itemize}

For the loss function $L_{\text{fid}}$, we calculate the entanglement fidelity $F_\text{ent}$ by choosing $\chi^a_i$ to be the weighted average of $\sum_{v}\bra{i_L} G^{(a)\dagger}_v G^{(a)}_m\ket{i_L}$ over $m$ given in \eqref{eqn:wt_avg_expression}.

We now apply this algorithm to obtain numerically optimized codes for GAD noise. To this end, group the error operators as described in Section \ref{subsec:encoding} and define notation for the grouped error operators, as follows.
 $\mathbb{E}_{\parallel}:= G_m^{(a) \dagger} G_n^{(a)}\ \forall\ a,m,n$, and,
    $\mathbb{E}_{\perp}:= G_m^{(a)\dagger} G_n^{(b)}\ \forall\ m,n \text{ and } a \neq b$.

\begin{algorithm}[H]
\caption{Computing Loss Functions (tol=$10^{-3}$)}
\label{alg:3}
\begin{algorithmic}[1]
    \STATE $L_{\text{PA-lap}} \gets 0$
    \STATE $\mathbb{E}_{\perp}, \mathbb{E}_{\parallel} \gets \text{Sample} (\gamma, p)$
    \newline
    \FOR{$|i_L\rangle \in \mathcal{C}$}
    \FOR{$G_i \in \mathbb{E}_{\parallel}$}
        \STATE $\chi \gets [\ ]$
        \newline
        \FOR{$E_n \in G_i$}
            \STATE $\chi[n] \gets \sum \langle i_L | E | i_L\rangle\ \forall\ E \in E_n$
        \ENDFOR
        \newline
        \STATE $\xi \gets \frac{\sum_n \chi}{\#\chi}$
        \STATE $L_{\text{PA-lap}} \gets L_{\text{PA-lap}} + ||\chi_a - \xi||^2 \ \forall \ \chi_a \in \chi $

    \ENDFOR
    \ENDFOR
    \newline

    \STATE $L_{\text{fid}} \gets - F_\text{ent}(\{\ |i_L\rangle\ \})$
    \STATE $L_{\text{PA-Norm}} \gets \epsilon \cdot \sum_i (1 - ||E \cdot i_L||)^2$
    \STATE $L_{\text{PA-Ortho}} \gets \sum_{i\neq j} ||\langle j_L | \mathbb{E}_{\parallel} | i_L \rangle||^2 + \sum_{i, j} ||\langle j_L | \mathbb{E}_{\perp} | i_L \rangle||^2$

    \STATE $L_{\text{Code-Norm}} \gets \sum_i (1 - ||\ |i_L\rangle\ ||)^2$
    \STATE $L_{\text{Code-Ortho}} \gets \sum_{i\neq j} ||\ \langle j_L | i_L \rangle\ ||^2$
    \newline
    \RETURN All losses
\end{algorithmic}
\end{algorithm}

Since our search is computationally expensive and inefficient, we have applied optimizations that allow us to search the space more efficiently while doing less work.

\subsubsection{Operator Sampling}
Our method is highly dependent on the exact error operators used. While we construct all our error operators as a function of $(\gamma, p)$, in practice our goal is to resist noise \textit{up to} $(\gamma, p)$ and not \textit{at} $(\gamma, p)$. Therefore we will have to factor in
error operators with $(\gamma', p') \in (0, \gamma] \times [0, p]$. Therefore, in each iteration, we select a random $(\gamma', p')$ as above and regenerate the Kraus operators, which are then used in the optimization problem.

While, at face value, this would be a very expensive method, to ensure generalizability, we can make a small modification that leverages efficient caching, enabling us to calculate the operators only once. Instead of randomly selecting a $\gamma' \in (0, \gamma]$, we select $\gamma'$ at discrete values, so in our case, an array of all of $\gamma' = 0.02n\ \because\ n\in [1, 15]$. Since there are only 15 such values and a few more for $p$, over thousands of iterations, we will have far more cache hits than misses.
\subsubsection{Double Descent (optional)}
We also run a second round of optimizations on a simplified version of the previous loss function, with the PAQEC conditions relaxed. The 2nd optimization is also run with a much smaller tolerance since we have almost found our encoding within the main loop. The secondary loop is just to apply small corrections to the code to maximize fidelity, which one may apply optionally.
Here, we keep only the orthonormalization conditions on the codewords and the fidelity maximization condition. All other conditions on or including the error operators are dropped.
The codewords obtained via the optimization method may not be unique; that is, different runs may yield codewords with similar performance. However, we provide one pair of codewords obtained from the aforementioned optimization method.

\begin{align}
    \ket{0_L} &=[
        -0.082,  0.463,  0.442,  0.052,  0.438,  0.046,  0.028, \nonumber\\
        & 0.003, 0.437,  0.036,  0.022,  0.001,  0.016,  0   , 0   ,  0   , 0.439, \nonumber\\
        &   0.022,  0.014, 0   ,  0.009, 0   , 0   ,  0   ,0.007, 0.   , 0   ,  0   , 0   ,  0   ,  0   ,  0  ]^T, \nonumber\\
    \ket{1_L} &= [
        0   , 0   , 0   , 0,  0   ,  0.004, -0.066, -0.105, 0.004, 0.031,\nonumber\\
        & 0.018,  0.04 ,  0.002,  0.015,  0.005, -0.442,  -0.004, 0.011,\nonumber\\
        & -0.001, -0.032,  0.001, 0   , -0.005, -0.443, -0.013, -0.062, \nonumber\\
        &-0.022,  0.434, -0.027,  0.443,  0.445,  0.027]^T.
\end{align}
    \end{document}